\documentclass[11pt]{article}
\usepackage{geometry}
\usepackage{amsmath, amssymb, amsthm}
\usepackage{graphicx}
\usepackage{bm}
\usepackage[linesnumbered,ruled]{algorithm2e}
\SetKwRepeat{Do}{do}{while}
\usepackage{fullpage}
\usepackage{enumerate}
\usepackage{algorithmic}
\newtheorem{theorem}{Theorem}
\newtheorem{lemma}{Lemma}

\newtheorem{definition}{Definition}
\newtheorem{proposition}{Proposition}

\DeclareMathOperator{\dist}{dist}
\DeclareMathOperator{\elm}{elm}
\DeclareMathOperator{\sgn}{sgn}

\title{On Connectedness of Solutions to Integer Linear Systems\thanks{The conference proceedings version has appeared in Proceedings of the 16th Annual International Conference on Combinatorial Optimization and Applications (COCOA2023), LNCS 14461, pages 421--433, 2023. }
\thanks{This work was supported 
by JSPS KAKENHI Grant Number JP20H05795.}}
\author{Takasugu Shigenobu\thanks{%
Graduate School of Mathematics, 
Kyushu University.
{\ttfamily shigenobu.takasugu.563@s.kyushu-u.ac.jp}}
\and 
Naoyuki Kamiyama\thanks{%
Institute of Mathematics for Industry, 
Kyushu University.
{\ttfamily kamiyama@imi.kyushu-u.ac.jp}}
}
\date{}

\begin{document}

\maketitle

\begin{abstract}
An integer linear system (ILS) is a linear system with integer constraints. 
The solution graph of an ILS is defined as an undirected graph defined on the set of feasible solutions to the ILS.
A pair of feasible solutions is connected by an edge in the solution graph if the Hamming distance between them is $1$.
We consider a property of the coefficient matrix of an ILS such that the solution graph is connected for any right-hand side vector.
Especially, we focus on the existence of an elimination ordering (EO) 
of a coefficient matrix, which is known as the sufficient condition for the connectedness of the solution graph for any right-hand side vector. 
That is, we consider the question whether the existence of an EO of the coefficient matrix is a necessary condition for the connectedness of the solution graph for any right-hand side vector.
We first prove that if a coefficient matrix has at least four rows and at least three columns, then the existence of an EO may not be a necessary condition.
Next, we prove that if a coefficient matrix has at most three rows or at most two columns, then the existence of an EO is a necessary condition.
\end{abstract} 

\section{Introduction}
An integer linear system (ILS) has an $ m \times n $ real coefficient matrix $A$, an $m$-dimensional real vector $b$, and a positive integer $d$.
In this case, a feasible solution of the ILS is an $n$-dimensional integer vector $x \in \{0,1,\dots,d\}^n$ such that $A x \ge b$.
The solution graph of an ILS is defined as an undirected graph defined on the set of feasible solutions to the ILS.
A pair of feasible solutions is connected by an edge in the solution graph 
if the Hamming distance between them is $1$.

We consider a property of the coefficient matrix of an ILS such that the solution graph is connected for any right-hand side vector.
Especially, we focus on the existence of an elimination ordering (EO) of a coefficient matrix, which is know as the sufficient condition for the connectedness of the solution graph for any right-hand side vector. 
(See Section \ref{section:Pre} for the definition of an EO.)
That is, we consider the question whether the existence of an EO of the coefficient matrix is a necessary condition for the connectedness of the solution graph for any right-hand side vector.

The results of this paper are summarized as follows.
We first prove that if a coefficient matrix has at least four rows and at least three columns, then the existence of an EO may not be a necessary condition (Theorem \ref{thm:main1}).
On the other hand, we also prove that if a coefficient matrix has at most three rows or at most two columns, then the existence of an EO is a necessary condition (Theorem \ref{thm:main2}).
In fact, we prove the contraposition of the statement. 
That is, we prove that 
if a coefficient matrix does not have an EO, then there exists an 
right-hand side vector such that the solution graph is not connected.

Kimura and Suzuki \cite{KIMURA202188} proved that if the coefficient matrix of an ILS has an EO, then the solution graph of the ILS is connected for any right-hand side vector. 
Precisely speaking, Kimura and Suzuki~\cite[Theorem 5.1]{KIMURA202188} proved that 
if the complexity index of the coefficient matrix of an ILS 
introduced by Kimura and Makino~\cite{KIMURA201667} is less than $1$, 
then the solution graph is connected. 
Furthermore, Kimura and Makino~\cite[Lemma~3]{KIMURA201667} proved that 
the complexity index is less than $1$ if and only if the coefficient matrix has an EO.
The complexity index of the coefficient matrix of an ILS
is a generalization of the 
complexity index for the Boolean satisfiability problem (SAT) 
introduced by 
Boros, Crama, Hammer, and Saks~\cite{doi:10.1137/S0097539792228629complexity},
and it depends only on the sign of the elements of the matrix.
It is known that 
we can determine whether the coefficient matrix of an ILS has an EO in polynomial time.

The connectedness of the solution graph of an ILS is closely related to a reconfiguration problem of the ILS.
A reconfiguration problem is a problem of finding a sequence of feasible solutions from the initial solution to the target solution 
(see, e.g.,~\cite{IDHPSUU11,N18}).
For ILS, the standard reconfiguration problems asks whether a given pair of feasible solutions to the ILS belong to the same connected component of the solution graph of the ILS.
Therefore, if the solution graph of an ILS is connected, then the answer is always YES.
Kimura and Suzuki~\cite{KIMURA202188} proved that 
computational complexity of the reconfiguration problem 
for the set of feasible solutions of 
an ILS has trichotomy.

An ILS is closely related to SAT.
An instance of SAT can be formulated by an ILS.
It is known that 
computational complexity of the reconfiguration problem of SAT has dichotomy~\cite{doi:10.1137/07070440X,Schwerdtfeger14}.

See Appendix for the proofs of the 
statements marked by $\star$.

\section{Preliminaries}\label{section:Pre}
In this paper, let $\mathbb{R}$ and $\mathbb{Z}_{>0}$ denote the sets of real numbers and positive integers, respectively.
For all integers $n \in \mathbb{Z}_{>0}$, we define $[n] := \{ 1, 2, \ldots , n \}$.
First, we formally define an integer linear system and its solution graph.
Throughout this paper, we fix a positive integer $d$. 
Define $D := \{0,1,\dots,d\}$. 
The set $D$ represents the domain of a variable in an integer 
linear system. 

\begin{definition}[Integer linear system]
    An integer linear system {\rm(ILS)} has a coefficient matrix $A=(a_{ij}) \in \mathbb{R}^{[m] \times [n]} $ and a vector $b \in \mathbb{R}^{[m]}$.
    This ILS is denoted by $I = (A,b)$.
    A feasible solution to $I$ is a vector $x \in D^{[n]}$ such that $A x \ge b$. 
    The set of feasible solutions to $I$ is denoted by $R(I)$ or $R(A,b)$.
\end{definition}

\begin{definition}[Hamming distance]
    Define the function $\dist \colon \mathbb{R}^{[n]} \times \mathbb{R}^{[n]} \to \mathbb{R}$ by 
    \(
    \dist(x,y) :=  \left| \left\{ j \in \{1,\ldots,n\} : x_j \neq y_j \right\} \right|
    \)
    for all vectors $x,y \in \mathbb{R}^{[n]}$. This function is called 
    the Hamming distance on $\mathbb{R}^{[n]}$.
\end{definition}

\begin{definition}[Solution graph]
    Let $R$ be a subset of $D^{[n]}$.
    We define the vertex set $V(R) := R$ and the edge set $E(R) := \{ \{ x, y \} : x,y \in V(R) , \dist(x, y) = 1 \}$.
    We define the solution graph $G(R)$ as the undirected graph with the vertex set $V(R)$ and the edge set $E(R)$.
    Furthermore, for each ILS $I$, we define $G(I) := G(R(I))$. 
\end{definition}

Next, we define elimination and an eliminated matrix.
These concepts are used to define an elimination ordering (EO).

\begin{definition}[Elimination]\label{def:elimination}
    Let $A = (a_{ij})$ be a matrix in $\mathbb{R}^{[m]\times [n]}$.
    Let $j$ be an integer in $[n]$.
    We say that $A$ can be eliminated at the column $j$ if it satisfies at least one of the following conditions.
    \begin{enumerate}[(i)]
        \item
        For all integers $i \in [m]$, if $ a_{ij} > 0$, then $a_{ij^{\prime}}=0$ for all integers $j^{\prime} \in [n] \setminus \{j\}$.
        \item
        For all integers $i \in [m]$, if $ a_{ij} < 0$, then $a_{ij^{\prime}}=0$ for all integers $j^{\prime} \in [n]\setminus \{j\}$.
    \end{enumerate}
\end{definition}

\begin{definition}[Eliminated matrix]
    Let $A$ be a matrix in $\mathbb{R}^{[m]\times [n]}$.
    Let $J$ be a subset of $[n]$.
    We define the eliminated matrix $\elm(A,J) \in \mathbb{R}^{ [m] \times ([n] \setminus J) }$ as the matrix obtained from $A$ by eliminating the $j$th column for all integers $j \in J$.
    We call the matrix $\elm(A,J)$ the eliminated matrix of $A$ by $J$.
\end{definition}

\begin{definition}[Elimination ordering]
    Let $A$ be a matrix in $\mathbb{R}^{[m]\times [n]}$.
    Let $S = (j_1,j_2,\dots, j_n)$ be a sequence of 
    integers in $[n]$.
    Then $S$ is called an elimination ordering {\rm(EO)} of $A$ if, for all integers $t \in [n]$, 
    $\elm(A,\{j_1,j_2,\dots,j_{t-1}\})$ can be eliminated at $j_t$.
\end{definition}

Finally, we define the sign function 
as follows.

\begin{definition}[Sign function]
    For all real numbers $x \in \mathbb{R}$, the sign function $\sgn : \mathbb{R} \to \{-1,0,1\}$ is 
    defined 
    as follows. 
    If $x < 0$, then we define $\sgn(x) := -1$. 
    If $x = 0$, then we define $\sgn(x) := 0$. 
    If $x > 0$, then we define $\sgn(x) := 1$. 
\end{definition}

\subsection{Our contribution}

In this paper, we prove following theorems.

\begin{theorem}\label{thm:main1}
    Suppose that $m \ge 4$ and $n \ge 3$. 
    Then there exists a matrix $A \in \mathbb{R}^{[m]\times [n]}$ satisfying 
    the following conditions.
    (i) $A$ does not have an EO.
    (ii) For all vectors $b \in \mathbb{R}^{[m]}$, 
    the solution graph $G(R(A,b))$ is connected.
\end{theorem}

\begin{theorem}\label{thm:main2}
    Let $A$ be a matrix in $\mathbb{R}^{[m]\times [n]}$.
    Suppose that, for all vectors $b \in \mathbb{R}^{[m]}$, the solution graph $G(R(A,b))$ is connected.
    Then if at least one of $m \le 3$ and $n \le 2$ holds, then 
    $A$ has an EO. 
\end{theorem}

\section{Proof of Theorem \ref{thm:main1}}
First, we prove following lemma.
This lemma plays an important role 
in the proof of Theorem~ \ref{thm:main1}.

\begin{lemma}\label{lem:m=4,n=3}
    There exists a matrix $A \in \mathbb{R}^{[4]\times [3]}$ satisfying 
    the following conditions.
    (i) The matrix $A$ does not have an EO.
    (ii) For all vectors $b \in \mathbb{R}^{[4]}$, 
    the solution graph $G(R(A,b))$ is connected.
\end{lemma}
\begin{proof} 
We define the matrix $A$ as follows.
\begin{equation}\label{eq:thm1}
    A = 
    \left(
    \begin{matrix}
        1 & 1 & 0 \\
        1 & -1 & 0\\
        -1 & 0 & 1 \\
        -1 & 0 & -1 
    \end{matrix}
    \right).
\end{equation}
It is not difficult to see that $A$ does not have an EO.

Suppose that $R(I)$ is not empty.
We take arbitrary vectors $b \in \mathbb{R}^{[4]}$ and $s,t \in R(I)$.
Without loss of generality, we suppose that $s_1 \ge t_1$.
We take the path $P$ from $s$ to $t$ defined by
\[
    s = 
    \begin{pmatrix}
        s_1 \\
        s_2 \\
        s_3
    \end{pmatrix}
    \to
    u^1 = 
    \begin{pmatrix}
        s_1 \\
        t_2 \\
        s_3
    \end{pmatrix}
    \to
    u^2 = 
    \begin{pmatrix}
        t_1 \\
        t_2 \\
        s_3
    \end{pmatrix}
    \to
    t = 
    \begin{pmatrix}
        t_1 \\
        t_2 \\
        t_3
    \end{pmatrix}.
\]
We prove that $u^1 \in R(I)$
because 
\begin{align*}
    s_1 + t_2 \ge t_1 + t_2 &\ge b_1, \ \ 
    s_1 - t_2 \ge t_1 - t_2 \ge b_2, \\
    - s_1 + s_3 &\ge b_3, \ \
    - s_1 - s_3 \ge b_4.
\end{align*}
We prove that $u^2 \in R(I)$
because 
\begin{align*}
    t_1 + t_2 &\ge b_1, \ \
    t_1 - t_2 \ge b_2, \\
    - t_1 + s_3 \ge - s_1 + s_3 &\ge b_3, \ \
    - t_1 - s_3 \ge - s_1 - s_3 \ge b_4.
\end{align*}
These imply that 
every vertex of $P$ is contained in 
$R(I)$. 
Thus, $G(I)$ is connected.
This completes the proof.
\end{proof}
\begin{proof}[Proof of Theorem \ref{thm:main1}]
Let $A$ be the matrix defined in \eqref{eq:thm1}.
If $m>4$ or $n>3$,
then we add rows and columns whose all elements are $0$ to $A$
until A becomes an $m \times n$ matrix.
Lemma \ref{lem:m=4,n=3} completes the proof.
\end{proof}

\section{Proof of Theorem \ref{thm:main2}}
First, we prove Lemma~\ref{lem:expansion_lemma}, which 
we call Expansion Lemma.
Expansion Lemma means that the columns which can be eliminated have nothing to do with the connectedness of the solution graph.

\subsection{Expansion Lemma}\label{section:expansion_lemma}
Let $A = (a_{ij})$ be a matrix in $\mathbb{R}^{[m] \times [n]}$.
Suppose that $A$ does not have an EO.
We define the subsets $\Delta,\,E \subseteq [n]$ as the output of Algorithm \ref{alg}.
We define the matrix $A^r = (a^r_{ij})$ as the submatrix of $A$ 
whose index set of columns is $\Delta$.
Similarly, we define the matrix $A^e = (a^e_{ij})$ as the submatrix of $A$ 
whose index set of columns is $E$.

\begin{algorithm}[h]
\caption{Algorithm for defining $\Delta$ and $E$.}
\label{alg}
\begin{algorithmic}[1]    
\STATE  $\Delta \leftarrow \emptyset$,~$E \leftarrow \emptyset$
\WHILE{$\elm(A,E)$ can be eliminated at some column}
\STATE Find an index $j \in [n] \setminus E$ at which the matrix $\elm(A,E)$ can be eliminated.
\STATE $E \leftarrow E \cup \{ j \}$ 
\ENDWHILE
\STATE $\Delta \leftarrow [n] \setminus E$
\STATE Output $\Delta$, $E$
\end{algorithmic}
\end{algorithm}

\begin{lemma}[Expansion Lemma]\label{lem:expansion_lemma}
    Suppose that there exists a vector $b^r \in \mathbb{R}^{[m]}$ 
    such that the solution graph $G(I^r)$ of the ILS  $I^r = (A^r,b^r)$ is not connected.
    Then there exists a vector $b \in \mathbb{R}^{[m]}$ 
    such that the solution graph $G(I)$ of the ILS  $I = (A,b)$ is not connected.
\end{lemma}

\begin{proof}
We define the vector $x^e \in D^{E}$ as follows.
\[
x^e_k := 
\left\{
\begin{aligned}
& 1 && \text{if the column vector $A_k$ is eliminated by the rule 
 (i) in Def.~\ref{def:elimination}} \\
& 0 && \text{otherwise}.
\end{aligned}
\right.
\]
With this vector $x^e$, we define the vector $b^e \in D^{[m]}$ by 
\[
b^e_i := \sum_{k \in E} a^e_{i k} d (1 - x^e_k) \quad (i \in [m]) .
\]

With the vectors $b^e$ and $b^r$, we define the vector $b$ by $b := b^r + b^e$.
We prove that the solution graph $G(I)$ of the ILS $I=(A,b)$ is not connected.
For each vector $z \in D^{\Delta}$ and each vector $\zeta \in D^{E}$, we define the vector $(z,\zeta) \in D^{[n]}$ by  
\[
(z,\zeta)_k :=
\left\{
\begin{aligned}
& z_k && (k \in \Delta) \\
& \zeta_k && (k \in E) .
\end{aligned}  
\right.
\]

\begin{proposition}\label{prop:EO1}
There exists a vector $\zeta^{\prime} \in D^{E}$ such that $(z,\zeta^{\prime}) \in R(I)$ for all feasible solutions $z \in R(I^r)$.
\end{proposition}
\begin{proof}
We define the vector $\zeta^{\prime} \in D^{E}$ by $\zeta^{\prime}_k := d(1-x^e_k)$
for all integers $k \in E$.

Since $z \in R(I^r)$, 
for all integers $i \in [m]$, we have $\sum_{k \in \Delta} a^r_{i k} z_{k} \ge b^r_i$.
Thus, for all integers $i \in [m]$, we have
\begin{align*}
\sum_{k \in [n]} a_{i k} (z,\zeta')_k - b_i 
&= \sum_{k \in \Delta} a^r_{i k} z_{k} - b^r_i
+ \sum_{k \in E} a^e_{i k} \zeta'_{k} - b^e_i \\
&\ge \sum_{k \in E} a^e_{i k} \zeta'_{k} - \sum_{k \in E} a^e_{ik} d (1 - x^e_{k}) = 0.
\end{align*}
Thus, we have $(z,\zeta^{\prime}) \in R(I)$.
This completes the proof.
\end{proof}

\begin{proposition}\label{prop:EO2}
    For all vectors $z \in D^{\Delta} \setminus R(I^r)$ and $\zeta \in D^{E}$, we have $(z,\zeta) \notin R(I)$.
\end{proposition}

\begin{proof}
Since $z$ is not a feasible solution in $R(I^r)$, 
there exists an integer $i \in [m]$ such that $\sum_{k \in \Delta} a^r_{i k} z_{k} < b^r_i$.

We prove that there exists an integer $j \in \Delta$ such that $a^r_{i j} \neq 0$.
Suppose that $a^r_{i j} = 0$ for all integers $j \in \Delta$.
For all vectors $z^{\prime} \in R(I^r)$,
$0 = \sum_{k \in \Delta} a^r_{i k} z^{\prime}_k \ge b^r_i$.
Therefore, we have 
$\sum_{k \in \Delta} a^r_{i k} z_{k} = 0 \ge b^r_i$.
It contradicts $\sum_{k \in \Delta} a^r_{i k} z_{k} < b^r_i$.
There exists an integer $j \in \Delta$ such that $a^r_{i j} \neq 0$.
We fix such an integer $j \in \Delta$.

We prove that, for all integers $k \in E$, $x^e_k = 1$ (resp. $x^e_k = 0$) 
implies $a^e_{i k} \le 0$ (resp. $a^e_{i k} \ge 0$).
Suppose that there exist an integer $k \in E$ such that $x^e_k = 1$ (resp. $x^e_k = 0$) 
and $a^e_{i k} > 0$ (resp. $a^e_{i k} < 0$).
Since $x^e_k = 1$ (resp. $x^e_k = 0$) and $k \in E$, 
the column vector $A_k$ is eliminated by the rule (i) (resp. (ii)) in 
Definition~\ref{def:elimination}.
Therefore, since $a^e_{i k} > 0$ (resp. $a^e_{i k} < 0$), for all integers 
$j^{\prime} \in [n] \setminus \{j\}$, we have $a_{i j^{\prime}} = 0$.
However, we have $a^r_{i j} \neq 0$ and $j \in \Delta \subseteq [n] \setminus \{j\}$.
This is a contradiction.
Thus, for all integers $k \in E$, $x^e_k = 1$ (resp. $x^e_k = 0$)  
implies $a^e_{i k} \le 0$ (resp. $a^e_{i k} \ge 0$).

We define $E^1$ (resp. $E^0$) as the set of integers $k \in E$ such that $x^e_k = 1$ (resp. $x^e_k = 0$).
For all integers $k \in E^1$, since $a^e_{i k} \le 0$, we have
\(
a^e_{i k} \zeta_{k} 
\le a^e_{i k} 0 
= a^e_{i k} d (1 - x^e_{k}).
\)
Similarly, for all integers $k \in E^0$, 
since $a^e_{i k} \ge 0$, we have
\(
a^e_{i k} \zeta_{k} 
\le a^e_{i k} d
= a^e_{i k} d (1 - x^e_{k}).
\)
Thus, for all integers $k \in E$, we have  $a^e_{i k} \zeta_{k} \le a^e_{i k} d (1 - x^e_{k})$.
We have
\begin{align*}
\sum_{k \in [n]} a_{i k} (z,\zeta)_k - b_i
&= \sum_{k \in \Delta} a^r_{i k} z_{k} - b^r_i
+\sum_{k \in E} a^e_{i k} \zeta_{k} - b^e_i 
< \sum_{k \in E} a^e_{i k} \zeta_{k} 
- b^e_i \\
&\le \sum_{k \in E} a^e_{i k} d (1 - x^e_{k})
- \sum_{k \in E} a^e_{i k} d (1 - x^e_{k}) =0,
\end{align*}
where 
the strict inequality follows from $\sum_{k \in \Delta} a^r_{i k} z_{k} < b^r_i$.
This completes the proof.
\end{proof}

We take vectors $p,q \in R(I^r)$ that are 
not connected on $G(I^r)$.
We take a vector $\zeta^{\prime} \in D^{E}$ satisfying the condition in
Proposition \ref{prop:EO1}.
We obtain $(p,\zeta^{\prime}), (q,\zeta^{\prime}) \in R(I)$.
We take an arbitrary path $P$ from $(p,\zeta^{\prime})$ to 
$(q,\zeta^{\prime})$ on $G(D^{[n]})$.
Let $(p,\zeta^{\prime}) = (u^{(0)},v^{(0)}) \to (u^{(1)},v^{(1)}) \to \cdots \to (u^{(\ell)},v^{(\ell)}) = (q,\zeta^{\prime})$ denote $P$.

Define the map $F^r : D^{[n]} \to D^{\Delta}$ by $F^r((z,\zeta)) := z$ for all vectors $(z,\zeta) \in D^{[n]}$.
Define the path $P^r$ as  
$F^r((u^{(0)},v^{(0)})) \to \cdots \to F^r((u^{(\ell)},v^{(\ell)}))$
on $G(D^{\Delta})$
($P^r$ may contain some duplicate vertices).
Since $p,q$ are not connected on $G(I^r)$, there exists a positive 
integer $k < \ell$ such that $F^r((u^{(k)},v^{(k)})) \notin R(I^r)$.

By Proposition \ref{prop:EO2}, $F^r((u^{(k)},v^{(k)})) \notin R(I^r)$ implies that,
for any vector $\zeta \in D^{E}$,
$(F^r((u^{(k)},v^{(k)})),\zeta) \notin R(I)$.
If we take $v^{(k)}$ as $\zeta$, then we have 
$(u^{(k)},v^{(k)}) = (F^r((u^{(k)},v^{(k)})),v^{(k)}) \notin R(I)$.
This implies that $P$ is not a path in $G(I)$.
Thus,
the solution graph $G(I)$ is not connected.
This completes the proof.
\end{proof}

\subsection{Two rows}
In this subsection, we consider the case where the coefficient matrix of an ILS consists of two rows.

\begin{proposition}[$\star$]\label{prop:sgna_1_eq_minus_sgna_2}
    Let $A=(a_{ij})$ be a matrix in $\mathbb{R}^{[2] \times [n]}$.
    Suppose that $A$ cannot be eliminated at any column. Then for all integers $j \in [n]$,
    $\sgn(a_{1 j}) = - \sgn(a_{2 j})$,
    $\sgn(a_{1 j}) \neq 0$, and 
    $\sgn(a_{2 j}) \neq 0$. 
\end{proposition}

\begin{lemma}\label{lemma:alldm=2}
    Let $A=(a_{ij})$ be a matrix in $\mathbb{R}^{[2] \times [n]}$.
    Suppose that $A$ cannot be eliminated at any column.
    Then there exists a vector $b \in \mathbb{R}^{[2]}$ such that 
    the solution graph $G(I)$ of the ILS  $I = (A,b)$ is not connected.
\end{lemma}
\begin{proof}
We define the set of integer $\{ j^1_1 , \ldots , j^1_n\} = [n]$ (resp.\ $\{j^2_1 , \ldots , j^2_n\} = [n]$) 
by $|a_{1 j^1_k}|\le|a_{1 j^1_{k+1}}|$ (resp.\ $|a_{2 j^2_k}|\le|a_{2 j^2_{k+1}}|$) for all integers $k \in [n-1]$.
That is, we arrange the elements in each row 
in non-decreasing order. 

Define the vector $x \in \{ 0, 1 \}^{[n]}$ by
\[
   x_k  := 
   \left\{
   \begin{aligned}
       &0& &(a_{1 k} > 0) \\
       &1& &(a_{1 k} < 0).
   \end{aligned}
   \right.
\]
Notice that Proposition \ref{prop:sgna_1_eq_minus_sgna_2}
implies that $a_k \neq 0$.

We define the vector $b \in \mathbb{R}^{[2]}$ by 
\[
\left(
\begin{aligned}
b_1 \\
b_2
\end{aligned}
\right)
:=
\left(
\begin{aligned}
& \sum_{k \in [n] \setminus \{ j^1_2\} } a_{1k} d (1 - x_k)
+ a_{1 j^1_2} ((d - 1) (1 - x_{j^1_2}) + x_{j^1_2}) \\
& \sum_{k \in [n] \setminus \{ j^2_1\} } a_{2k} d (1 - x_k)
+ a_{2 j^2_1} ((d - 1) (1 - x_{j^2_1}) + x_{j^2_1})
\end{aligned}
\right) .
\]
Then we consider the ILS $I = (A,b)$.
We define the vectors $p,q \in D^{[n]}$ as follows. 
\begin{align*}
p_k :=
\left \{
\begin{aligned}
& d (1 - x_k) && (k \neq j^1_1) \\
& (d - 1) (1 - x_{k}) + x_{k} && ( k = j^1_1 ) ,
\end{aligned}
\right. 
q_k :=
\left \{
\begin{aligned}
& d (1 - x_k) && (k \neq j^1_2) \\
& (d - 1) (1 - x_{k}) + x_{k} && ( k = j^1_2 ) .
\end{aligned}
\right.
\end{align*}
We prove that the vectors $p,q$ belong to $R(I)$ and 
they are not connected on the solution graph $G(I)$.

\begin{proposition}[$\star$] \label{prop:pq_in_R(I)_ver2row}
The vectors $p,q$ belong to $R(I)$.
\end{proposition}

\begin{proposition}[$\star$] \label{prop:a(1-2s)=|a|_2row}
    For all integers $j \in [n]$, $a_{1 j}(1 - 2 x_{j}) = |a_{1 j}|$ and
    $a_{2 j}(1 - 2 x_{j}) = - |a_{2 j}|$.
\end{proposition}

We define $Y$ as the set of vectors $y \in D^{[n]}$ 
such that $\dist(q,y) = 1$.
In other words, the subset $Y$ is the set of 
neighborhood vertices of $q$ on $G(D^{[n]})$.
Then we prove that $y \notin R(I)$ for all vectors $y \in Y$.

We arbitrarily take a vector $y \in Y$.
From the definition, the following equation is obtained for the vector $y$.
\[
    y_k = 
    \left\{
    \begin{aligned}
    & q_k && (k \neq j) \\
    & \xi && (k = j) ,
    \end{aligned}
    \right.
\]
where $j$ is an integer in $[n]$ and $\xi$ is an integer in $D$ such that $\xi \neq q_j$.
See Appendix~\ref{appendix:diagram}
for the case distinction of the following proof.

\noindent \textbf{Case 1 ($j \neq j^1_2$).}
If $j \neq j^1_2$, then we have
\begin{align*}
\sum_{k\in [n]} a_{1 k} y_k - b_1
&= \sum_{k \in [n] \setminus \{ j \}}  a_{1k} q_k + a_{1j}\xi - b_1 \\ 
&= \sum_{k \in [n] \setminus \{ j,j^1_2 \}}a_{1k} d (1 - x_k)
+ a_{1 j^1_2} ((d-1)(1-x_{j^1_2})+x_{j^1_2}) + a_{1j} \xi  \\
&\quad - \left( \sum_{k \in [n] \setminus \{ j^1_2\} } a_{1k} d (1 - x_k)
+ a_{1 j^1_2} ((d - 1) (1 - x_{j^1_2}) + x_{j^1_2}) \right) \\
&=a_{1 j}(\xi - d(1-x_j)).
\end{align*}

\noindent \textbf{Case 1.1 ($a_{1j} > 0$).}
If $a_{1 j}>0$, then $x_j = 0$.
Since $q_j = d(1-x_j) = d$, the inequality 
$0 \le \xi \le d-1$
is obtained.
We have
\[
    \sum_{k\in [n]} a_{1 k} y_k - b_1 =
    a_{1 j}(\xi - d(1-x_j)) 
    \le  a_{1 j}((d - 1) - d) = - a_{1 j} < 0. 
\]

\noindent \textbf{Case 1.2 ($a_{1j} < 0$).}
If $a_{1 j}<0$, then $x_j = 1$.
Since $q_j = d(1-x_j) = 0$, the inequality 
$1 \le \xi \le d$
is obtained.
We have
\[
    \sum_{k\in [n]} a_{1 k} y_k - b_1 = 
    a_{1 j}(\xi - d(1-x_j)) 
    \le a_{1 j}(1 - 0) = a_{1 j} < 0 .
\]

\noindent \textbf{Case 2 ($j = j^1_2$).}
If $j = j^1_2$, then we have
\begin{align*}
\sum_{k\in [n]} a_{1 k} y_k - b_1 
&= \sum_{k \in [n] \setminus \{j\}} a_{1 k} q_k + a_{1 j} \xi \\
&\quad - \left( \sum_{k \in [n] \setminus \{ j\} } a_{1k} d (1 - x_k)
+ a_{1 j} ((d - 1) (1 - x_{j}) + x_{j}) \right) \\
&=a_{1 j}(\xi - ((d - 1) (1 - x_{j}) + x_{j} )).
\end{align*}

\noindent \textbf{Case 2.1 ($a_{1 j}>0$).}
If $a_{1 j}>0$, then $x_{j} = 0$.
Thus, $q_j=(d - 1) (1 - x_{j}) + x_{j} = d-1$.
Therefore, either $0 \le \xi \le d-2 $ or $\xi=d$ is satisfied.
We have
\begin{align*}
\sum_{k\in [n]} a_{1 k} y_k - b_1 
= a_{1 j}(\xi -( (d - 1) (1 - x_{j}) + x_{j} ) )
= a_{1 j}(\xi - (d - 1)).
\end{align*}

\noindent \textbf{Case 2.1.1 ($0 \le \xi \le d-2$).}
If $0 \le \xi \le d-2$, then we have
\begin{align*}
\sum_{k\in [n]} a_{1 k} y_k - b_1 
= a_{1 j}(\xi - (d - 1)) 
\le a_{1 j}((d-2) - (d - 1)) = - a_{1 j} < 0 .
\end{align*}

\noindent \textbf{Case 2.1.2 ($\xi=d$).}
If $\xi=d$, then 
we consider $b_2$.

\noindent \textbf{Case 2.1.2.1 ($j \neq j^2_1$).}
By Proposition \ref{prop:a(1-2s)=|a|_2row},
if $j \neq j^2_1$, then we have
\begin{align*}
&\sum_{k\in [n]} a_{2 k} y_k - b_2 \\
&= \sum_{k \in [n]} a_{2k} y_k
- \left( \sum_{k \in [n] \setminus \{ j^2_1\} } a_{2k} d (1 - x_k)
+ a_{2 j^2_1} ((d - 1) (1 - x_{j^2_1}) + x_{j^2_1}) \right) \\
&= a_{2 j}(\xi - d(1 - x_j) ) 
+ a_{2 j^2_1}(y_{j^2_1} - ((d - 1) (1 - x_{j^2_1}) + x_{j^2_1} ))  \\
&= a_{2 j}(d - d(1 - 0) ) 
+ a_{2 j^2_1}(q_{j^2_1} - ((d - 1) (1 - x_{j^2_1}) + x_{j^2_1} ) ) \\
&= a_{2 j^2_1}(d(1 - x_{j^2_1}) - ((d - 1) (1 - x_{j^2_1}) + x_{j^2_1} ) ) \\
&= a_{2 j^2_1}(1-2x_{j^2_1}) = -|a_{2 j^2_1}| < 0 .
\end{align*}

\noindent \textbf{Case 2.1.2.2 ($j = j^2_1$).}
If $j = j^2_1$, then 
$a_{2 j} < 0$ follows from $a_{1 j} > 0$
and Proposition \ref{prop:sgna_1_eq_minus_sgna_2}.
\begin{align*}
&\sum_{k\in [n]} a_{2 k} y_k - b_2 \\
&= \sum_{k \in [n]} a_{2k} y_k
- \left( \sum_{k \in [n] \setminus \{ j\} } a_{2k} d (1 - x_k)
+ a_{2 j} ((d - 1) (1 - x_{j}) + x_{j}) \right) \\
&= a_{2 j}(\xi - ((d - 1) (1 - x_{j}) + x_{j} ))  \\
&= a_{2 j}(d - ((d - 1) (1 - 0) + 0 ))  = a_{2 j} < 0 .
\end{align*}
See Appendix~\ref{appendix:remaining} 
for the remaining part of the proof.

We obtain $y \notin R(I)$ for all the cases.
Therefore, any neighborhood vertex 
of $q$ on $G(D^{[n]})$ is not a feasible solution.
This completes the proof. 
\end{proof}

\begin{lemma}\label{lem:alldm=2EO}
    Let $A = (a_{ij})$ be a matrix in $\mathbb{R}^{[2] \times [n]}$.
    Suppose that $A$ does not have an EO.
    Then there exists a vector $b \in \mathbb{R}^{[2]}$ such that the solution graph $G(I)$ of the ILS  $I = (A,b)$ is not connected.
\end{lemma}

\begin{proof}
We define $A^r$ in the same way as in Section~\ref{section:expansion_lemma}. 
Then $A^r$ cannot be eliminated at any column.
By Lemma \ref{lemma:alldm=2}, there exists a 
vector $b^r$ such that the solution graph $G(R(A^r,b^r))$ is not connected.
Lemma \ref{lem:expansion_lemma} completes the proof.
\end{proof}

\subsection{Two columns}
In this subsection, we consider the case where the coefficient matrix of an ILS consists of two columns.
We prove the following lemma.

\begin{lemma}[$\star$] \label{lem:alldn1=2EO}
    Let $A = (a_{ij})$ be a matrix in $\mathbb{R}^{[m] \times [2]}$.
    Suppose that $A$ does not have an EO.
    Then there exists a vector $b \in \mathbb{R}^{[m]}$ such that the solution graph $G(I)$ of the ILS  $I = (A,b)$ is not connected.
\end{lemma}

\subsection{Three rows}
In this subsection, we consider the case where the coefficient matrix of an ILS consists of three rows.
We prove the following lemma.
At the end of this section, we prove Theorem \ref{thm:main2}.

\begin{lemma}\label{lem:m=3}
    Let $A = (a_{ij})$ be a matrix in $\mathbb{R}^{[3] \times [n]}$.
    Suppose that $A$ cannot be eliminated at any column.
    Then there exists a vector 
    $b \in \mathbb{R}^{[3]}$ such that the solution graph $G(I)$ of 
    the ILS $I = (A,b)$ is not connected.
\end{lemma}

\begin{proof}
For all integers $i_1,i_2 \in [3]$, we define 
\[
\Lambda_{i_1,i_2} := \{ j \in [n] : \sgn(a_{i_1 j}) = - \sgn(a_{i_2 j}) \neq 0 \}.
\] 

\begin{proposition}[$\star$] \label{prop:row3-1}
    If there exist integers $i_1,i_2 \in [3]$ such that 
    $|\Lambda_{i_1,i_2}| \ge 2$, then for all integers 
    $d \in \mathbb{Z}_{>0}$, there exists a vector 
    $b \in \mathbb{R}^{[3]}$ such that the solution graph $G(I)$ of 
    the ILS $I = (A,b)$ is not connected.
\end{proposition}

\begin{proposition}[$\star$]\label{prop:row3-2}
    If $n \neq 3$, then 
    there exist integers $i_1,i_2 \in [3]$ such that $|\Lambda_{i_1,i_2}| \ge 2$.
\end{proposition}

\begin{proposition}[$\star$]\label{prop:row3-3}
    Suppose that $n=3$ and there are distinct integers $i^{\prime}_1,i^{\prime}_2,i^{\prime}_3 \in [3]$
    such that $\Lambda_{i^{\prime}_1,i^{\prime}_2} \cap \Lambda_{i^{\prime}_2,i^{\prime}_3} \neq 0$.
    Then there exist integers $i_1,i_2 \in [3]$ such that 
    $|\Lambda_{i_1,i_2}| \ge 2$.
\end{proposition}

Proposition \ref{prop:row3-1} implies that 
if there exist integers $i_1,i_2 \in [3]$ such that $|\Lambda_{i_1,i_2}| \ge 2$, then 
the proof is done.
Suppose that, for all integers $i_1, i_2 \in [3]$, $|\Lambda_{i_1,i_2}| = 1$.
By Proposition \ref{prop:row3-2}, we have $n = 3$.
By Proposition \ref{prop:row3-3}, 
$\Lambda_{1,2}$, $\Lambda_{2,3}$, and $\Lambda_{1,3}$ are pairwise disjoint. 
Notice that, for all integers $j \in [3]$, 
there exist integers $i_1, i_2 \in [3]$ such that $j \in \Lambda_{i_1,i_2}$.
Without loss of generality
$\Lambda_{1,2} = \{1\}$,
$\Lambda_{2,3} = \{2\}$, and
$\Lambda_{1,3} = \{3\}$.
We have $a_{1 2} = a_{2 3} = a_{3 1} = 0$.
For example, if $a_{12} \neq 0$, then $2 \in \Lambda_{1,2}$ or 
$2 \in \Lambda_{1,3}$.

We define the vector $x \in \{0,1\}^{[3]}$ by 
\[
x_k :=
\left\{
\begin{aligned}
    & 0 && \mbox{if $a_{k k} > 0$}\\ 
    & 1 && \mbox{if $a_{k k} < 0$}
\end{aligned}
\right.
\qquad (k \in [3]) .
\]

For all integers $i \in [3]$, 
we assume that 
$\{ j^i_1 , j^i_2 , j^i_3\} = [3]$ and 
$|a_{i j^i_1}| \le |a_{i j^i_2}| \le |a_{i j^i_3}|$.
By the definition, for all integers $i \in [3]$, we have $a_{i j^i_1} =0$.
We define the vector $b \in \mathbb{R}^{[3]}$ by
\[
b_i := 
\left\{
\begin{aligned}
    & \sum_{k \in [n]} a_{i k} d (1 - x_k)
    &&\mbox{if $j^i_2 = i$} \\
    & \sum_{k \in [n]} a_{i k} ((d - 1) (1 - x_k) + x_k)
    &&\mbox {if $j^i_2 \neq i$}
\end{aligned}
\right.
\qquad ( i \in [3] ) .
\]
Then we consider the ILS $I = (A,b)$.

We define the vectors $p,q \in D^{[3]}$ as follows.
\begin{align*}
     p_i &:= d (1 - x_i) & (i \in [3]), \\
     q_i &:= (d - 1) (1 - x_i) + x_i &  (i \in [3]).
\end{align*}
We prove that the vectors $p,q$ belong to $R(I)$ and
they are not connected
on the solution graph $G(I)$.

\begin{proposition}[$\star$]\label{prop:pq_in _R(I)_Ver_row3}
    The vectors $p,q$ belong to $R(I)$.
\end{proposition}

\begin{proposition}[$\star$]\label{prop:a(1-2s)=|a|_3}
    For all integers $j \in [3]$, we have $a_{j j}(1 - 2 x_{j}) = |a_{j j}|$.
    For all integers $i \in [3]$ and all integers $s \in \{ 2 , 3 \}$,
    if $j^i_s \neq i$, then we have $a_{i j^i_s}(1 - 2 x_{j^i_s}) = - |a_{i j^i_s}|$.
\end{proposition}

We define $Y$ as the set of vectors $y \in D^{[n]}$ 
such that $\dist(p,y) = 1$.
In other words, the subset $Y$ is the set of 
neighborhood vertices of $p$ on $G(D^{[n]})$.
Then we prove that $y \notin R(I)$ for all vectors $y \in Y$.

We arbitrarily take a vector $y \in Y$.
From the definition, the following equation is obtained for the vector $y$.
\[
    y_k = 
    \left\{
    \begin{aligned}
    & p_k && (k \neq \ell) \\
    & \xi && (k = \ell) .
    \end{aligned}
    \right.
\]
where $\ell$ is an integer in $[3]$ and $\xi$ is an integer in $D$ such that $\xi \neq p_{\ell}$.

Suppose that $j^{\ell}_2 = \ell$.
For all vectors $y \in Y$, we have 
\begin{align*}
    \sum_{k \in [3]} a_{\ell k} y_k - b_{\ell} 
    = \sum_{k \in [3]} a_{\ell k} y_k -  \sum_{k \in [3]} a_{\ell k} d (1 - x_{\ell}) 
    = a_{\ell \ell}(\xi - d (1 - x_{\ell})).
\end{align*}

If $a_{\ell \ell} > 0$, then $x_{\ell} = 0$.
Since $p_{\ell} = d (1 - x_{\ell}) = d$, the inequality 
$0 \le \xi \le d-1 $ is obtained. 
We have
\begin{align*}
    \sum_{k \in [3]} a_{\ell k} y_k - b_{\ell} 
    = a_{\ell \ell}(\xi - d (1 - x_{\ell}))
    \le a_{\ell \ell} (d - 1 - d) 
    = - a_{\ell \ell} < 0.
\end{align*}

If $a_{\ell \ell} < 0$, then $x_{\ell} = 1$.
Since $p_{\ell} = d (1 - x_{\ell}) = 0$, the inequality 
$1 \le \xi \le d$ is obtained. 
We have
\begin{align*}
    \sum_{k \in [3]} a_{\ell k} y_k - b_{\ell} 
    = a_{\ell \ell}(\xi - d (1 - x_{\ell}))
    \le a_{\ell \ell} < 0.
\end{align*}

Suppose that $j^{\ell}_2 \neq \ell$.
By Proposition \ref{prop:a(1-2s)=|a|_3}, for all vectors $y \in Y$, we have 
\begin{align*}
    \sum_{k \in [3]} a_{\ell k} y_k - b_{\ell}
    &= a_{\ell j^{\ell}_2}d(1 -x_{j^{\ell}_2}) + a_{\ell \ell} \xi -b_\ell \\
    &= a_{\ell j^{\ell}_2}(d ( 1 - x_{j^{\ell}_2}) 
    - ((d-1)(1-x_{j^{\ell}_2}) + x_{j^{\ell}_2}) \\
    &\qquad + a_{\ell \ell}(\xi - ((d-1)(1-x_{\ell}) + x_{\ell})) \\
    &= a_{\ell j^{\ell}_2}(1 - 2 x_{j^{\ell}_2}) 
    + a_{\ell \ell}(\xi - ((d-1)(1-x_{\ell}) + x_{\ell})) \\
    &= -|a_{\ell j^{\ell}_2}| + 
    a_{\ell \ell}(\xi - ((d-1)(1-x_{\ell}) + x_{\ell})) .
\end{align*}

If $a_{\ell \ell} > 0$, then $x_{\ell} = 0$.
Since $p_{\ell} = d (1 - x_{\ell}) = d$, the inequality 
$0 \le \xi \le d-1$ is obtained. 
\begin{align*}
    \sum_{k \in [3]} a_{\ell k} y_k - b_{\ell}  
    &= -|a_{\ell j^{\ell}_2}| + 
    a_{\ell \ell}(\xi - ((d-1)(1-x_{\ell}) + x_{\ell})) \\
    &\le -|a_{\ell j^{\ell}_2}| + a_{\ell \ell}(d -1 - (d - 1))
    = -|a_{\ell j^{\ell}_2}| < 0 .
\end{align*}

If $a_{\ell \ell} < 0$, then $x_{\ell} = 1$.
Since $p_{\ell} = d (1 - x_{\ell}) = 0$, the inequality 
$1 \le \xi \le d$ is obtained. 
We have
\begin{align*}
    \sum_{k \in [3]} a_{\ell k} y_k - b_{\ell} 
    &= -|a_{\ell j^{\ell}_2}| + 
    a_{\ell \ell}(\xi - ((d-1)(1-x_{\ell}) + x_{\ell})) \\
    &\le -|a_{\ell j^{\ell}_2}| + a_{\ell \ell}(1 - 1)
    = -|a_{\ell j^{\ell}_2}| < 0 .
\end{align*}

For all vectors $y \in Y$, we obtain $y \in R(I)$.
Therefore, any neighborhood vertex of $p$ on $G(D^{[n]})$ is not a feasible solution.
This completes the proof.
\end{proof}

\begin{lemma}\label{lem:m=3EO}
    Let $A = (a_{ij})$ be a matrix in $\mathbb{R}^{[3] \times [n]}$.
    Suppose that $A$ does not have an EO.
    Then there exists a vector 
    $b \in \mathbb{R}^{[3]}$ such that the solution graph $G(I)$ of 
    the ILS $I = (A,b)$ is not connected.
\end{lemma}
\begin{proof}
We define $A^r$ in the same way as in Section~\ref{section:expansion_lemma}. 
Then $A^r$ cannot be eliminated at any column.
By Lemma \ref{lem:m=3}, we have the vector $b^r$ such that the solution graph $G(R(A^r,b^r))$ is not connected.
Lemma \ref{lem:expansion_lemma} completes the proof.
\end{proof}

\begin{proof}[Proof of Theorem \ref{thm:main2}]
We consider the contraposition of the statement in Theorem \ref{thm:main2}.
If $m=2$ (resp. $n=2$, $m=3$), Lemma \ref{lem:alldm=2EO} (resp. Lemma \ref{lem:alldn1=2EO}, Lemma \ref{lem:m=3EO}) completes this proof.
\end{proof}

\bibliographystyle{plain}
\bibliography{ILS_connectivity}

\clearpage 

\appendix
\section{Omitted Figures}
\subsection{Branch diagram for the proof of Lemma~\ref{lemma:alldm=2}}
\label{appendix:diagram}
\begin{figure}[h]
\centering
\includegraphics[width=14cm]{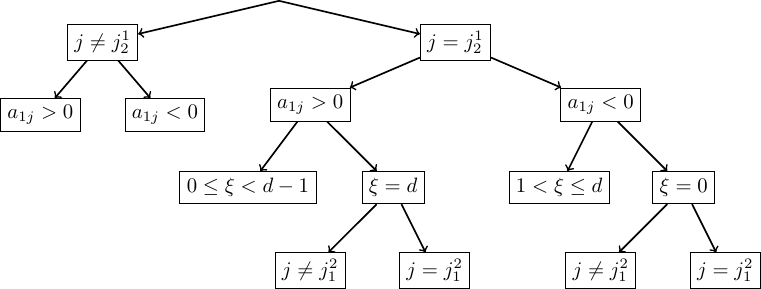}
\caption{The branch diagram for the proof of Lemma \ref{lemma:alldm=2}.}
\label{fig:lemma:alldm=2}
\end{figure}

\section{Omitted Proofs}

\subsection{Proof of Proposition~\ref{prop:sgna_1_eq_minus_sgna_2}}
\begin{proof}
    Since $A$ cannot be eliminated at any column, $A$ has at least one positive element and at least one negative element in each column.
    Since $A$ has two rows, any column of $A$ does not have $0$ as its element.
    This completes the proof.
\end{proof}

\subsection{Proof of Proposition~\ref{prop:pq_in_R(I)_ver2row}}
\begin{proof}
We use Proposition \ref{prop:a(1-2s)=|a|_2row} in this proof.
Notice that Proposition~\ref{prop:pq_in_R(I)_ver2row} is not used
in the proof of Proposition \ref{prop:a(1-2s)=|a|_2row},

First, we consider $p$.
By Proposition \ref{prop:a(1-2s)=|a|_2row},
we have
\begin{align*}
    &\sum_{k\in [n]} a_{1 k} p_k - b_1 \\ 
    &= \sum_{k \in [n]} a_{1k} p_k
    - \left( \sum_{k \in [n] \setminus \{ j^1_2\} } a_{1k} d ( 1 - x_k)
    + a_{1 j^1_2} ((d - 1) (1 - x_{j^1_2}) + x_{j^1_2}) \right) \\
    &= a_{1 j^1_1}(p_{j^1_1} - d(1 - x_{j^1_1}) ) 
    + a_{1 j^1_2}(p_{j^1_2} - ((d - 1) (1 - x_{j^1_2}) + x_{j^1_2} ))  \\
    &= a_{1 j^1_1}(((d - 1) (1 - x_{j^1_1}) + x_{j^1_1})- d(1-x_{j^1_1}) ) \\
    & \qquad + a_{1 j^1_2}(d (1 - x_{j^1_2}) - ((d - 1) (1 - x_{j^1_2}) + x_{j^1_2} ))  \\
    &= - a_{1 j^1_1}(1 - 2 x_{j^1_1}) 
    + a_{1 j^1_2}(1 - 2 x_{j^1_2}) = - |a_{1 j^1_1}| + |a_{1 j^1_2}| \ge 0 .
\end{align*}

Next, we consider the element $b_2$. If $j^1_1 = j^2_1$, then $\sum_{k\in [n]}a_{2 k}p_k=b_2$.
By Proposition \ref{prop:a(1-2s)=|a|_2row},
if $j^1_1 \neq j^2_1$, then we have
\begin{align*}
    &\sum_{k\in [n]} a_{2 k} p_k - b_2 \\
    &= \sum_{k \in [n]} a_{2k} p_k
    - \left( \sum_{k \in [n] \setminus \{ j^2_1\} } a_{2k} d (1 - x_k)
    + a_{2 j^2_1} ((d - 1) (1 - x_{j^2_1}) + x_{j^2_1}) \right) \\
    &= a_{2 j^1_1}(p_{j^1_1} - d(1 - x_{j^1_1}) ) 
    + a_{2 j^2_1}(p_{j^2_1} - ((d - 1) (1 - x_{j^2_1}) + x_{j^2_1} )  )\\
    &= a_{2 j^1_1}((d - 1) (1 - x_{j^1_1}) + x_{j^1_1} - d(1-x_{j^1_1}) ) \\
    & \qquad + a_{2 j^2_1}(d (1 - x_{j^2_1}) -( (d - 1) (1 - x_{j^2_1}) + x_{j^2_1} )  )\\
    &= - a_{2 j^1_1}(1 - 2 x_{j^1_1}) + a_{2 j^2_1}(1 - 2 x_{j^2_1}) 
    = |a_{2 j^1_1}| - |a_{2 j^2_1}| \ge 0 .
\end{align*}

Similarly, for $q$, we have
\begin{align*}
    &\sum_{k\in [n]} a_{1 k} q_k - b_1 \\
    &= \sum_{k \in [n]} a_{1k} q_k
    - \left( \sum_{k \in [n] \setminus \{ j^1_2\} } a_{1k} d (1 - x_k)
    + a_{1 j^1_2} ((d - 1) (1 - x_{j^1_2}) + x_{j^1_2}) \right) = 0.
\end{align*}
We consider the element $b_2$. If $j^1_2 = j^2_1$, then $\sum_{k\in [n]}a_{2 k}q_k=b_2$.
By Proposition \ref{prop:a(1-2s)=|a|_2row},
if $j^1_2 \neq j^2_1$, then we have
\begin{align*}
    &\sum_{k\in [n]} a_{2 k} q_k - b_2 \\
    &= \sum_{k \in [n]} a_{2k} q_k
    - \left( \sum_{k \in [n] \setminus \{ j^2_1\} } a_{2k} d (1 - x_k)
    + a_{2 j^2_1} ((d - 1) (1 - x_{j^2_1}) + x_{j^2_1}) \right) \\
    &= a_{2 j^1_2}(q_{j^1_2} - d(1 - x_{j^1_2}) ) 
    + a_{2 j^2_1}(q_{j^2_1} - ((d - 1) (1 - x_{j^2_1}) + x_{j^2_1} ) ) \\
    &= a_{2 j^1_2}(((d - 1) (1 - x_{j^1_2}) + x_{j^1_2}) - d(1-x_{j^1_2}) ) \\
    & \qquad + a_{2 j^2_1}(d (1 - x_{j^2_1}) - ( (d - 1) (1 - x_{j^2_1}) + x_{j^2_1} ) ) \\
    &= - a_{2 j^1_2}(1 - 2 x_{j^1_2}) + a_{2 j^2_1}(1 - 2 x_{j^2_1}) 
    = |a_{2 j^1_2}| - |a_{2 j^2_1}| \ge 0 .
\end{align*}
This completes the proof.
\end{proof}

\subsection{Proof of Proposition~\ref{prop:a(1-2s)=|a|_2row}}
\begin{proposition}\label{prop:1-2x=sgna1j}
    For all integers $j \in [n]$, $1 - 2 x_{j} = \sgn(a_{1 j}) = -\sgn(a_{2 j})$.
\end{proposition}
\begin{proof}
If $a_{1 j} > 0$, then $x_{j} = 0$ and $\sgn(a_{1 j})=1$.
If $a_{1 j} < 0$, then $x_{j} = 1$ and $\sgn(a_{1 j})=0$.
Since $x_{j} = 0$ (resp.\ $x_{j} = 1$) 
implies $1 - 2 x_{j}=1$ (resp.\ $1 - 2 x_{j}=0$), 
we have $1 - 2 x_{j} = \sgn(a_{1 j})$.
The second equation follows from Proposition \ref{prop:sgna_1_eq_minus_sgna_2}.
\end{proof}

\begin{proof}[Proof of Proposition~\ref{prop:a(1-2s)=|a|_2row}]
By Proposition \ref{prop:1-2x=sgna1j}, we have 
\[
a_{1 j}(1 - 2 x_{j}) 
= |a_{1 j}|\sgn(a_{1 j})(1 - 2 x_{j}) 
= |a_{1 j}| \sgn(a_{1 j})^2 = |a_{1 j}|.
\]
Similarly, since $1 - 2 x_{j} = - \sgn(a_{2 j})$, we have 
\[
a_{2 j}(1 - 2 x_{j}) 
= |a_{2 j}|\sgn(a_{2 j})(1 - 2 x_{j}) 
= - |a_{2 j}| \sgn(a_{2 j})^2 = - |a_{2 j}|.
\]
This completes proof.
\end{proof}

\subsection{The remaining part of the proof of Lemma~\ref{lemma:alldm=2}}
\label{appendix:remaining}

\noindent \textbf{Case 2.2 ($a_{1 j}<0$).}
If $a_{1 j}<0$, then $x_{j} = 1$.
Thus, $q_j=(d - 1) (1 - x_{j}) + x_{j} = 1$.
Therefore, either $2 \le \xi \le d$ or $\xi=0$ is satisfied.
We have
\begin{align*}
\sum_{k\in [n]} a_{1 k} y_k - b_1 
= a_{1 j}(\xi - ((d - 1) (1 - x_{j}) + x_{j} )  )
= a_{1 j}(\xi - 1) .
\end{align*}

\noindent \textbf{Case 2.2.1 ($2 \le \xi \le d$).}
If $2 \le \xi \le d$, then we have
\begin{align*}
\sum_{k\in [n]} a_{1 k} y_k - b_1 
= a_{1 j}(\xi - 1) 
\le a_{1 j}(2 - 1) = a_{1 j} < 0 . 
\end{align*}

\noindent \textbf{Case 2.2.2 ($\xi=0$).}
In this case, we consider $b_2$.

\noindent \textbf{Case 2.2.2.1 ($j \neq j^2_1$).}
By Proposition \ref{prop:a(1-2s)=|a|_2row},
if $j \neq j^2_1$, then we have
\begin{align*}
&\sum_{k\in [n]} a_{2 k} y_k - b_2 \\
&= \sum_{k \in [n]} a_{2k} y_k
- \left( \sum_{k \in [n] \setminus \{ j^2_1\} } a_{2k} d (1 - x_k)
+ a_{2 j^2_1} ((d - 1) (1 - x_{j^2_1}) + x_{j^2_1}) \right) \\
&= a_{2 j}(\xi - d(1 - x_j) ) 
+ a_{2 j^2_1}(y_{j^2_1} - ((d - 1) (1 - x_{j^2_1}) + x_{j^2_1} ))  \\
&= a_{2 j}(0 - d(1 - 1) ) 
+ a_{2 j^2_1}(q_{j^2_1} - ((d - 1) (1 - x_{j^2_1}) + x_{j^2_1} ) ) \\
&= a_{2 j^2_1}(d(1 - x_{j^2_1}) - ((d - 1) (1 - x_{j^2_1}) + x_{j^2_1} ) ) 
= a_{2 j^2_1}(1-2x_{j^2_1}) 
= -|a_{2 j^2_1}| < 0.
\end{align*}

\noindent \textbf{Case 2.2.2.2 ($j = j^2_1$).}
If $j = j^2_1$, then $a_{2 j} > 0$ follows from $a_{1 j} < 0$ 
and Proposition \ref{prop:sgna_1_eq_minus_sgna_2}.
Thus, 
\begin{align*}
&\sum_{k\in [n]} a_{2 k} y_k - b_2 \\
&= \sum_{k \in [n]} a_{2k} y_k
- \left( \sum_{k \in [n] \setminus \{ j\} } a_{2k} d (1 - x_k)
+ a_{2 j} ((d - 1) (1 - x_{j}) + x_{j}) \right) \\
&= a_{2 j}(\xi - ((d - 1) (1 - x_{j}) + x_{j} ))  
= a_{2 j}(0 - ((d - 1) (1 - 1) + 1 ))  
= - a_{2 j} < 0 .
\end{align*}

\subsection{Proof of Lemma~\ref{lem:alldn1=2EO}}

\begin{proof}
Since $A$ dose not have an EO, it cannot be eliminated at any column.

Here we prove the following proposition.
\begin{proposition}\label{lemma:alldn=2}
    There exist integers $i_1,i_2 \in [m]$ such that, 
    for any integer $j \in [2]$, 
    $\sgn(a_{i_1 j})=-\sgn(a_{i_2 j})$.
\end{proposition}
\begin{proof}
Since $A$ cannot be eliminated at any column, there exist integers $p,q,r \in [m]$ such that the following conditions.
\begin{enumerate}
    \item $a_{p1} \neq 0$ and $a_{p2} \neq 0$.
    \item $\sgn(a_{q1}) = -\sgn(a_{p1})$ and $a_{q2} \neq 0$.
    \item $\sgn(a_{r2}) = -\sgn(a_{p2})$ and $a_{r1} \neq 0$.
\end{enumerate}

First, we suppose that $\sgn(a_{p2}) = \sgn(a_{q2})$ and $\sgn(a_{p1}) = \sgn(a_{r1})$.
Then the above conditions 
imply that 
\[
\sgn(a_{q1}) = -\sgn(a_{p1}) = -\sgn(a_{r1}), \quad
\sgn(a_{q2}) = \sgn(a_{p2}) = -\sgn(a_{r2}).
\]
Thus, if we define $i_1 := q$ and $i_2 := r$, then we obtain $\sgn(a_{i_1 j})=-\sgn(a_{i_2 j})$
for any integer $j \in [2]$.

Next, we suppose that $\sgn(a_{q2}) = -\sgn(a_{p2})$. If we define $i_1 := p$ and $i_2 := q$, then, for any $j \in [2]$, we obtain $\sgn(a_{i_1 j})=-\sgn(a_{i_2 j})$.

Finally, we suppose that $\sgn(a_{r1}) = -\sgn(a_{p1})$.
If we define $i_1 := p$ and $i_2 := r$, then, for any $j \in [2]$, we obtain $\sgn(a_{i_1 j})=-\sgn(a_{i_2 j})$.
\end{proof}

We take integers $i_1,i_2 \in [m]$ satisfying the condition in
Proposition \ref{lemma:alldn=2}.
Using this, we define the submatrix $A^{\prime}$ of $A$ by
\[
A^{\prime} := 
\begin{pmatrix}
a_{i_1 1} & a_{i_1 2} \\
a_{i_2 1} & a_{i_2 2}
\end{pmatrix}.
\]

From Lemma \ref{lem:alldm=2EO}, there exists a vector 
$b^{\prime} \in \mathbb{R}^{[2]}$ such that the solution graph $G(I^{\prime})$ of the ILS $I^{\prime} = (A^{\prime}, b^{\prime})$ is not connected.
Using $b^{\prime}$, we define $b$ by
\[
b_k := \left\{
\begin{aligned}
&b'_1 && (k=i_1) \\
&b'_2 &&(k=i_2) \\
&- d \cdot \textstyle{\sum^{n}_{\ell=1} |a^r_{k \ell}|} && (k\neq i_1 , i_2)
.\end{aligned}
\right.
\]
Then it is not difficult to see that the solution graph $G(I)$ of the ILS $I=(A,b)$
is not connected.
This completes the proof.
\end{proof}

\subsection{Proof of Proposition~\ref{prop:row3-1}}
\begin{proof}
Suppose that 
there exist integers $i_1,i_2 \in [3]$ such that 
$|\Lambda_{i_1,i_2}| \ge 2$. 
Without loss of generality, we suppose that 
$i_1 = 1$ and $i_2 = 2$. 
We define the submatrix $A^{\prime}$ of $A$ by 
\[
A^{\prime} := 
\begin{pmatrix}
a_{1 1} & a_{1 2} & \cdots & a_{1 n}\\
a_{2 1} & a_{2 2} & \cdots & a_{2 n}
\end{pmatrix}.
\]

Since $|\Lambda_{1,2}| \ge 2$, $A^{\prime}$ does not have an EO.
Thus, we can apply Lemma \ref{lem:alldm=2EO} to $A^{\prime}$.
That is there exists a vector $b^{\prime} \in \mathbb{R}^{[2]}$ such that
the solution graph $G(I^{\prime})$ 
of the ILS $I^{\prime} = (A^{\prime},b^{\prime})$ is not connected.
Using $b^{\prime}$, we define the vector $b \in \mathbb{Z}^{[3]}$ by
\[
b :=
\begin{pmatrix}
    b^{\prime}_1 \\
    b^{\prime}_2 \\
    - d \cdot \sum^{n}_{k=1} |a^r_{3 k}|
\end{pmatrix} .
\]
Then the solution graph $G(I)$ of the ILS $I=(A,b)$
is not connected.
This completes the proof.
\end{proof}

\subsection{Proof of Proposition~\ref{prop:row3-2}}
\begin{proof}
Suppose that $n = 2$.
Proposition \ref{lemma:alldn=2} completes the proof.

Suppose that $n \ge 4$.
By the definition of elimination, each column vector of $A$ has at least one positive element and at least one negative element.
Each column vector of $A$ belongs to at least one of $\Lambda_{1,2}$, $\Lambda_{2,3}$, or $\Lambda_{1,3}$.
Since $n \ge 4$, at least one of $\Lambda_{1,2}$, $\Lambda_{2,3}$, and $\Lambda_{1,3}$ has at least two elements.
This completes the proof.
\end{proof}

\subsection{Proof of Proposition~\ref{prop:row3-3}}
\begin{proof}
Let $\ell_1,\ell_2$, and $\ell_3$ denote the elements of $[n]$.
Suppose that $\ell_1 \in \Lambda_{i^{\prime}_1,i^{\prime}_2} 
\cap \Lambda_{i^{\prime}_2,i^{\prime}_3}$.

By the definition of elimination, each column vector of $A$ has at least one positive element and at least one negative element.
Each column vector of $A$ belongs to at least one of 
$\Lambda_{i^{\prime}_1,i^{\prime}_2}$, 
$\Lambda_{i^{\prime}_2,i^{\prime}_3}$, and
$\Lambda_{i^{\prime}_1,i^{\prime}_3}$.

Suppose that at least one of $\ell_2, \ell_3$ belongs to $\Lambda_{i^{\prime}_1,i^{\prime}_2}$ or 
$\Lambda_{i^{\prime}_2,i^{\prime}_3}$.
At least one of $\Lambda_{i^{\prime}_1,i^{\prime}_2}$ and $\Lambda_{i^{\prime}_2,i^{\prime}_3}$
has two elements.
If both $\ell_2$ and $\ell_3$ do not belong to $\Lambda_{i^{\prime}_1,i^{\prime}_2}$ or 
$\Lambda_{i^{\prime}_2,i^{\prime}_3}$, 
then $\ell_2, \ell_3 \in \Lambda_{i^{\prime}_1,i^{\prime}_3}$.
This completes the proof.
\end{proof}

\subsection{Proof of Proposition~\ref{prop:pq_in _R(I)_Ver_row3}}
\begin{proof}
We use Proposition \ref{prop:a(1-2s)=|a|_3} in this proof.
Notice that Proposition~\ref{prop:pq_in _R(I)_Ver_row3} is not used
in the proof of Proposition \ref{prop:a(1-2s)=|a|_3}.

We consider $p$. 
We take arbitrary integer $i \in [3]$.
If $j^i_2 = i$, then $\sum_{k \in [3]} a_{i k} p_k = b_i$.
Suppose that $j^i_2 \neq i$.
In this case, $i = j^i_3$.
By Proposition \ref{prop:a(1-2s)=|a|_3},
we have 
\begin{align*}
    \sum_{k \in [3]} a_{i k} p_k - b_i
    &= \sum_{k \in [3]} a_{i k} d (1 - x_k) - \sum_{k \in [n]} a_{i k} ((d - 1) (1 - x_k) + x_k) \\
    &= \sum_{k \in [3]} a_{i k} (1 - 2 x_k) 
    = a_{i i}(1 - 2 x_i) + a_{i j^i_2} (1 - 2 x_{j^i_2}) \\
    &= |a_{i i}| - |a_{i j^i_2}| \ge 0 .
\end{align*}
This completes the proof that $p$ belongs to $R(I)$.

We consider $q$. 
We take arbitrary integer $i \in [3]$.
If $j^i_2 \neq i$, then $\sum_{k \in [3]} a_{i k} q_k = b_i$.
Suppose that $j^i_2 = i$.
In this case, $j^i_3 \neq i$.
By Proposition \ref{prop:a(1-2s)=|a|_3},
we have 
\begin{align*}
    \sum_{k \in [3]} a_{i k} q_k - b_i 
    &= \sum_{k \in [3]} a_{i k} ((d - 1) (1 - x_k) + x_k)
    - \sum_{k \in [n]} a_{i k} d (1 - x_k) \\
    &= \sum_{k \in [3]} - a_{i k} (1 - 2 x_k) 
    = - a_{i i}(1 - 2 x_i) - a_{i j^i_3} (1 - 2 x_{j^i_3}) \\
    &= - |a_{i i}| + |a_{i j^i_3}| \ge 0.
\end{align*}
This completes the proof.
\end{proof}

\subsection{Proof of Proposition~\ref{prop:a(1-2s)=|a|_3}}
\begin{proposition}\label{prop:1-2x=sgnajj}
    For all integers $j \in [3]$, we have $1 - 2 x_j = \sgn(a_{j j})$.
\end{proposition}
\begin{proof}
If $a_{j j} > 0$, then $x_{j} = 0$ and $\sgn(a_{j j})=1$.
If $a_{j j} < 0$, then $x_{j} = 1$ and $\sgn(a_{j j})=0$.
Since $x_{j} = 0$ (resp.\ $x_{j} = 1$) 
implies $1 - 2 x_{j}=1$ (resp.\ $1 - 2 x_{j}=0$), 
we have $1 - 2 x_{j} = \sgn(a_{j j})$.
\end{proof}

\begin{proposition}\label{prop:1-2x=sgnajj2}
    For all integers $i \in [3]$ and all integers $s \in \{ 2 , 3 \}$
    if $j^i_s \neq i$, then we have $1 - 2 x_{j^i_s} = - \sgn(a_{i j^i_s})$.
\end{proposition}
\begin{proof}
We prove $\sgn(a_{j^i_s j^i_s}) = - \sgn(a_{i j^i_s})$.
Since $s \neq 1$, we have $ a_{i j^i_s} \neq 0$.
Suppose that $\sgn(a_{j^i_s j^i_s}) = \sgn(a_{i j^i_s})$.
Since $A$ cannot be eliminated at any column, 
we have $\sgn(a_{t j^i_s}) = - \sgn(a_{j^i_s j^i_s}) = - \sgn(a_{i j^i_s})$
, where $t \in [3] \setminus \{j^i_s,i\}$.
Therefore, we have $j^i_s \in \Lambda_{t,j^i_s} \cap \Lambda_{t,i}$.
Proposition \ref{prop:row3-3} that there exist integers $i_1,i_2 \in [3]$ such that $|\Lambda_{i_1,i_2}| \ge 2$.
It contradicts the assumption that $|\Lambda_{i_1,i_2}| = 1$.
Thus, we have $\sgn(a_{j^i_s j^i_s}) = - \sgn(a_{i j^i_s})$.

By Proposition \ref{prop:1-2x=sgnajj}, we have $1 - 2 x_{j^i_s} = \sgn(a_{j^i_s j^i_s}) = - \sgn(a_{i j^i_s})$.
This completes the proof.
\end{proof}

\begin{proof}[Proof of Proposition~\ref{prop:a(1-2s)=|a|_3}]
By Proposition \ref{prop:1-2x=sgnajj}, for all integers $j \in [3]$, we have
\[
a_{j j}(1 - 2 x_{j}) 
= |a_{j j}|\sgn(a_{j j})(1 - 2 x_{j}) 
= |a_{j j}| \sgn(a_{j j})^2 = |a_{j j}|.
\]
By Proposition \ref{prop:1-2x=sgnajj2}, for all integers $i \in [3]$ and all integers $s \in \{ 2 , 3 \}$, we have 
\[
a_{i j^i_s}(1 - 2 x_{j^i_s}) 
= |a_{i j^i_s}|\sgn(a_{i j^i_s})(1 - 2 x_{j^i_s}) 
= - |a_{i j^i_s}| \sgn(a_{i j^i_s})^2 = - |a_{i j^i_s}|.
\]
This completes proof.
\end{proof}

\end{document}